\documentclass[letterpaper,journal]{IEEEtran}
\usepackage{latexsym}
\usepackage{amsmath}
\usepackage{amssymb}
\usepackage{color}
\usepackage{bm}
\usepackage[dvipdf]{graphicx}
\usepackage{graphicx}
\usepackage{enumerate}
\usepackage{amsmath}
\usepackage{amssymb}
\usepackage{amsfonts}
\usepackage{amsthm}
\usepackage{bm}

\newcommand{\hHD}{\hat{H}_D}
\newcommand{\hHC}{\hat{H}_C}
\newcommand{\HD}{{H}_{\Delta}}
\newcommand{\HG}{{H}_{\Gamma}}

\newcommand{\Z}{{\mathbb Z}}

\newcommand{\GF}{\mathrm{GF}}

\newcommand{\tentative}{\noindent{{\underline{\tt tentative decision}}} }

\newcommand{\ctov}{\noindent{{\underline{\tt horizontal step}}} }
\newcommand{\vtoc}{\noindent{{\underline{\tt vertical step}}} }

\newcommand{\initialization}{\noindent{{\underline{\tt initialization}}} }

\newtheorem{proposition}{Proposition}
\newtheorem{theorem}[proposition]{Theorem}

\newtheorem{lemma}[proposition]{Lemma}
\newtheorem{fact}[proposition]{Fact}
\newtheorem{observation}[proposition]{Observation}
\newtheorem{example}[proposition]{Example}

\title{Quantum Error Correction beyond the Bounded Distance Decoding Limit}
\author{
Kenta~Kasai,~\IEEEmembership{Member,~IEEE,} 
Manabu~Hagiwara,~\IEEEmembership{Member,~IEEE,} 
Hideki~Imai,~\IEEEmembership{Life Fellow,~IEEE,} 
and~Kohichi~Sakaniwa,~\IEEEmembership{Senior Member,~IEEE}
}
\begin{document}
\maketitle
\begin{abstract}
In this paper, we consider  quantum error correction over depolarizing channels 
with non-binary low-density parity-check codes defined over Galois field of size $2^p$. 
The proposed quantum error correcting codes are 
based on the binary quasi-cyclic CSS (Calderbank, Shor and Steane) codes. 
The resulting quantum codes outperform the best known quantum codes and 
surpass the performance limit of the bounded distance decoder. 
By increasing the size of the underlying Galois field, i.e., $2^p$, the error floors are considerably improved. 
\end{abstract}
%%%%%%%%%%%%%%%%%%%%%%%%%%%%%%%%%%%%%%%%%%
\section{Introduction}
%%%%%%%%%%%%%%%%%%%%%%%%%%%%%%%%%%%%%%%%%%
\begin{figure*}[t]
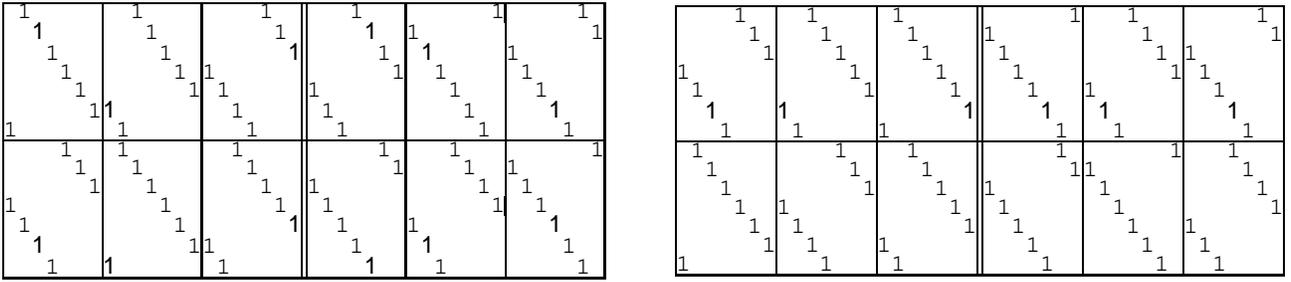

\begin{center}
 \begin{tabular}{cc}
 \begin{minipage}{0.47\textwidth}
\small{
\renewcommand\arraystretch{0.5}
\begin{tabular}
{|@{}l@{}@{}l@{}@{}l@{}@{}l@{}@{}l@{}@{}l@{}@{}l@{}|@{}l@{}@{}l@{}@{}l@{}@{}l@{}@{}l@{}@{}l@{}@{}l@{}|@{}l@{}@{}l@{}@{}l@{}@{}l@{}@{}l@{}@{}l@{}@{}l@{}||@{}l@{}@{}l@{}@{}l@{}@{}l@{}@{}l@{}@{}l@{}@{}l@{}|@{}l@{}@{}l@{}@{}l@{}@{}l@{}@{}l@{}@{}l@{}@{}l@{}|@{}l@{}@{}l@{}@{}l@{}@{}l@{}@{}l@{}@{}l@{}@{}l@{}|@{}}
\hline
&\tt{1}&&&&&&&&\tt{1}&&&&&&&&&\tt{1}&&&&&&\tt{1}&&&&&&&&&&\tt{1}&&&&&&\tt{1}&\\
&&\sf{1}&&&&&&&&\tt{1}&&&&&&&&&\tt{1}&&&&&&\sf{1}&&&\tt{1}&&&&&&&&&&&&&\tt{1}\\
&&&\tt{1}&&&&&&&&\tt{1}&&&&&&&&&\sf{1}&&&&&&\tt{1}&&&\sf{1}&&&&&&\tt{1}&&&&&&\\
&&&&\tt{1}&&&&&&&&\tt{1}&&\tt{1}&&&&&&&&&&&&&\tt{1}&&&\tt{1}&&&&&&\tt{1}&&&&&\\
&&&&&\tt{1}&&&&&&&&\tt{1}&&\tt{1}&&&&&&\tt{1}&&&&&&&&&&\tt{1}&&&&&&\tt{1}&&&&\\
&&&&&&\tt{1}&\sf{1}&&&&&&&&&\tt{1}&&&&&&\tt{1}&&&&&&&&&&\tt{1}&&&&&&\sf{1}&&&\\
\tt{1}&&&&&&&&\tt{1}&&&&&&&&&\tt{1}&&&&&&\tt{1}&&&&&&&&&&\tt{1}&&&&&&\tt{1}&&\\\hline
&&&&\tt{1}&&&&\tt{1}&&&&&&&&\tt{1}&&&&&&&&&&\tt{1}&&&&&\tt{1}&&&&&&&&&&\tt{1}\\
&&&&&\tt{1}&&&&\tt{1}&&&&&&&&\tt{1}&&&&&&&&&&\tt{1}&&&&&\tt{1}&&&\tt{1}&&&&&&\\
&&&&&&\tt{1}&&&&\tt{1}&&&&&&&&\tt{1}&&&\tt{1}&&&&&&&&&&&&\tt{1}&&&\tt{1}&&&&&\\
\tt{1}&&&&&&&&&&&\tt{1}&&&&&&&&\tt{1}&&&\tt{1}&&&&&&&&&&&&\tt{1}&&&\tt{1}&&&&\\
&\tt{1}&&&&&&&&&&&\tt{1}&&&&&&&&\sf{1}&&&\tt{1}&&&&&\tt{1}&&&&&&&&&&\sf{1}&&&\\
&&\sf{1}&&&&&&&&&&&\tt{1}&\tt{1}&&&&&&&&&&\tt{1}&&&&&\sf{1}&&&&&&&&&&\tt{1}&&\\
&&&\tt{1}&&&&\sf{1}&&&&&&&&\tt{1}&&&&&&&&&&\sf{1}&&&&&\tt{1}&&&&&&&&&&\tt{1}&\\\hline
\end{tabular}
\renewcommand\arraystretch{1.0}
} \end{minipage}
 & \begin{minipage}{0.47\textwidth}
 \small{ 
\renewcommand\arraystretch{0.539}
\begin{tabular}
{|@{}l@{}@{}l@{}@{}l@{}@{}l@{}@{}l@{}@{}l@{}@{}l@{}|@{}l@{}@{}l@{}@{}l@{}@{}l@{}@{}l@{}@{}l@{}@{}l@{}|@{}l@{}@{}l@{}@{}l@{}@{}l@{}@{}l@{}@{}l@{}@{}l@{}||@{}l@{}@{}l@{}@{}l@{}@{}l@{}@{}l@{}@{}l@{}@{}l@{}|@{}l@{}@{}l@{}@{}l@{}@{}l@{}@{}l@{}@{}l@{}@{}l@{}|@{}l@{}@{}l@{}@{}l@{}@{}l@{}@{}l@{}@{}l@{}@{}l@{}|@{}}
\hline
&&&&\tt{1}&&&&&\tt{1}&&&&&&\tt{1}&&&&&&&&&&&&\tt{1}&&&&\tt{1}&&&&&&&&&\tt{1}&\\
&&&&&\tt{1}&&&&&\tt{1}&&&&&&\tt{1}&&&&&\tt{1}&&&&&&&&&&&\tt{1}&&&&&&&&&\tt{1}\\
&&&&&&\tt{1}&&&&&\tt{1}&&&&&&\tt{1}&&&&&\tt{1}&&&&&&&&&&&\tt{1}&&\tt{1}&&&&&&\\
\tt{1}&&&&&&&&&&&&\tt{1}&&&&&&\tt{1}&&&&&\tt{1}&&&&&&&&&&&\tt{1}&&\tt{1}&&&&&\\
&\tt{1}&&&&&&&&&&&&\tt{1}&&&&&&\tt{1}&&&&&\tt{1}&&&&\tt{1}&&&&&&&&&\tt{1}&&&&\\
&&\sf{1}&&&&&\sf{1}&&&&&&&&&&&&&\sf{1}&&&&&\sf{1}&&&&\sf{1}&&&&&&&&&\sf{1}&&&\\
&&&\tt{1}&&&&&\tt{1}&&&&&&\tt{1}&&&&&&&&&&&&\tt{1}&&&&\tt{1}&&&&&&&&&\tt{1}&&\\\hline
&\tt{1}&&&&&&&&&&\tt{1}&&&&&\tt{1}&&&&&&&&&&\tt{1}&&&&&&&&\tt{1}&&&&\tt{1}&&&\\
&&\tt{1}&&&&&&&&&&\tt{1}&&&&&\tt{1}&&&&&&&&&&\tt{1}&\tt{1}&&&&&&&&&&&\tt{1}&&\\
&&&\tt{1}&&&&&&&&&&\tt{1}&&&&&\tt{1}&&&\tt{1}&&&&&&&&\tt{1}&&&&&&&&&&&\tt{1}&\\
&&&&\tt{1}&&&\tt{1}&&&&&&&&&&&&\tt{1}&&&\tt{1}&&&&&&&&\tt{1}&&&&&&&&&&&\tt{1}\\
&&&&&\tt{1}&&&\tt{1}&&&&&&&&&&&&\tt{1}&&&\tt{1}&&&&&&&&\tt{1}&&&&\tt{1}&&&&&&\\
&&&&&&\tt{1}&&&\tt{1}&&&&&\tt{1}&&&&&&&&&&\tt{1}&&&&&&&&\tt{1}&&&&\tt{1}&&&&&\\
\tt{1}&&&&&&&&&&\tt{1}&&&&&\tt{1}&&&&&&&&&&\tt{1}&&&&&&&&\tt{1}&&&&\tt{1}&&&&\\\hline
\end{tabular}
 }\end{minipage}
 \end{tabular}
\end{center}
\label{120458_5Jul10}
\caption{An example of binary $(J=2, L=6, P=7)$-QC parity-check matrix pair $(\hHC, \hHD)$ constructed by the method in Theorem \ref{Hagiwara-Imai} with $\sigma=2$ and $\tau=3$. 
It holds that $\hHC\hHD^\mathsf{T}=0$. For any row $m'$ of $\hHD$
}
\label{134801_6Jul10}
\end{figure*}
In 1963, Gallager invented low-density parity-check (LDPC) codes \cite{gallager_LDPC}.
Due to the sparseness of the code representation, 
LDPC codes are efficiently decoded by the sum-product algorithm. 
By a powerful optimization method {\it density evolution} \cite{910577}, developed by Richardson and Urbanke,
messages of sum-product decoding can be statistically evaluated. 
The optimized LDPC codes exhibit error correcting performance very close to the Shannon limit \cite{richardson01design}.
Recently, LDPC codes have been generalized from a point of view of Galois fields, i.e. non-binary LDPC codes are proposed.
Non-binary LDPC codes were invented by Gallager \cite{gallager_LDPC}. 
Davey and MacKay \cite{DaveyMacKayGFq} found that non-binary LDPC codes can outperform binary ones. 

Quantum LDPC codes, which are quantum error-correcting codes, have been developed in a similar manner to (classical) LDPC codes.
By the discovery of CSS (Calderbank, Shor and Steane)
codes \cite{calderbank96, steane96b} and stabilizer codes \cite{gottesman96}, 
the notion of parity-check measurement, which is a generalized notion of parity-check matrix, is introduced to quantum information theory.
In particular, a parity-check measurement for a CSS code is characterized by a pair of parity-check matrices 
which satisfy the following condition: the product of one of the pair and the transposed other is subjected to be a zero-matrix.

Quantum LDPC codes were first introduced by Postol in \cite{Pos01a}. 
The above CSS constraint on the parity-check matrices makes the design of the quantum LDPC codes difficult. 
MacKay et al. proposed the {\it bicycle} codes \cite{1337106} and Cayley graph based CSS codes \cite{MoreSparseGraph}. 
In \cite{4957637}, Poulin et al. proposed serial turbo codes for the quantum error correction. These codes can be decoded by an efficient iterative decoder.
To the best of the authors' knowledge, these codes 
\cite{1337106, MoreSparseGraph, 4957637} are the best known quantum error correcting codes among efficiently decodable quantum LDPC codes so far. 
%It is known \cite{richardson01design} that the decoding performance of LDPC codes largely depend on the column and row weight profile of parity-check matrices. 
In \cite{4557323}, Hagiwara and Imai proposed a construction method of CSS code pair that has quasi-cyclic (QC) parity-check matrices with
arbitrary regular even row weight $L\ge 4$ and column weight $J$ such that $L/2\ge J\ge 2$. 
However, the resulting codes do not outperform the codes proposed by MacKay el al. \cite{1337106, MoreSparseGraph}. 
%This is probably due to the lack of  randomness in the quasi-cyclic parity-check matrices compared to the codes in \cite{1337106, MoreSparseGraph}. 

Generally, LDPC CSS codes tend to have poor minimum distance. 
The minimum distance of an LDPC CSS code is upper-bounded by the row weight of the parity-check matrix. 
This is  due to the dual and sparse constraint on the parity-check matrices. 
When the LDPC CSS codes are used with large code length, the poor minimum distance leads to high error floors. 
Therefore, it is desired to establish the construction method of quantum LDPC codes with large minimum distance. 
We should note that it is important to study quantum LDPC codes with large minimum distance which grows with code length  \cite{5205648} for constructing quantum LDPC codes with
vanishing decoding error probability. 
%, although this is beyond the scope of this paper. %different from our main research interest.

%\cite{Pos01a, Kit03a, 1337106, COT05a, COT07a, LG06a, LG08a}%, HH07a, Djo08a, SRK08a, Aly07b, Aly08a, HBD08a, TL08a}

Non-binary LDPC codes are defined as codes over $\GF(2^p)$ with $p>2$.  
The  parity-check matrices of non-binary LDPC codes are given as sparse matrices over $\GF(2^p)$. 
In this paper, we investigate non-binary LDPC codes for quantum error correction. 
It is empirically known that the best classical non-binary LDPC codes have column weight $J=2$
from a point of view of error-correcting performance \cite{4641893}.
Moreover, due to the sparse representation of non-binary parity-check matrices of column weight $J=2$, 
the non-binary LDPC codes are efficiently decoded by FFT-based sum-product algorithm \cite{706440}. 

In this paper, we propose a construction method of a binary CSS code which can be viewed also as a pair of non-binary LDPC codes.
More precisely, the proposed construction method produces a binary code pair $(C, D)$ such that $C\supset D^{\perp}$, and 
$C$ and $D$ are also defined by non-binary sparse parity-check matrices over $\GF(2^p)$ of column weight $J=2$. 
This satisfies the constraint of CSS codes. 
To this end, we first construct $PJ\times PL$ binary QC parity-check matrix pair $(\hHC, \hHD)$ 
with column weight $J=2$ and row weight $L$ such that $\hHC\hHD^{\mathsf{T}}=0$ by the method developed in \cite{4557323}. 
Solving some linear equations on $\Z_{2^p-1}$, we get $PJ\times PL$ non-binary parity-check matrix pair $(\HG, \HD)$ 
with column weight $J=2$ and row weight $L$ such that $\HG\HD^{\mathsf{T}}=0$. 
It is known that a natural linear map from $\mathrm{GF}(2^p)$ to $\mathrm{GF}(2)^{p\times p}$ is given so that 
through this map, the non-binary LDPC matrix pair $(\HG, \HD)$ can be viewed as a binary LDPC matrix pair $(H_C, H_D)$ such that $H_CH_D^{\mathsf{T}}=0$.  
Numerical experiments show the resulting CSS codes outperform the best known quantum error correcting codes and surpass the performance limit of the bounded 
distance decoder. 
By increasing the size of the underlying Galois field, i.e., $2^p$, the error floors are considerably improved. 

% The proposed method is based on the construction  for obtaining a pair of parity-check matrices which satisfy the condition for CSS codes.
% a constant column (resp. row) weight 2 (resp. 4) over $\mathrm{GF}(2^m)$, and length $N$, where $m$ and $N$ are positive integers.
% By a natural linear map from $\mathrm{GF}(2^m)$ to $\mathrm{GF}(2)$ via the primitive polynomial, we coincide these ingredient matrices with matrices over $\mathrm{GF}(2)$.
% Therefore, the proposed quantum codes are defined over $(\mathbb{C}^{2})^{\otimes mN}$, in other words, quantum states are of level 2.
% From the ingredient matrices, we also have two pairs of classical linear codes which are defined over $\mathrm{GF}(2)$ and $\mathrm{GF}(2^m)$.
The rest of this paper is organized as follows. 
Section \ref{230852_10Jun10} describes the construction method of  a non-binary twisted LDPC parity-check matrix pair $(HG, \HD)$ of column weight $J=2$. 
%Section \ref{223910_10Jun10} explains how to represent a non-binary LDPC parity-check matrix pair $(\HG, \HD)$ as a binary parity-check matrix pair $(H_C, H_D)$, i.e., a CSS code. 
Section \ref{231512_10Jun10} describes the decoding algorithm of the binary twisted code pair $(C, D)$. 
Section \ref{231523_10Jun10} demonstrates the decoding performance of the proposed codes. 
%%%%%%%%%%%%%%%%%%%%%%%%%%%%%%%%%%%%%%%%%%%%%%%%%%%%%%%%%%%%%%%%%%%%%%%%%%%%%%%%
\section{Construction of Non-Binary Matrix Pair with Column Weight 2}
\label{230852_10Jun10}
%%%%%%%%%%%%%%%%%%%%%%%%%%%%%%%%%%%%%%%%%%%%%%%%%%%%%%%%%%%%%%%%%%%%%%%%%%%%%%%%
In this section, we will construct a binary code pair $(C, D)$ defined by orthogonal parity-check matrices $H_C$ and $H_D$, 
where we call two matrices $X$ and $Y$  orthogonal if $XY^\mathsf{T}=0$.
Let $pN$ qubits be the code length. 
The binary codes $C$ and $D$ are designed in such a way that $C$ and $D$ can be represented 
also by non-binary sparse parity-check matrices over $\GF(2^p)$ of column size $N$ and column weight $J=2$. 
To this end, we start with binary QC matrices and extend them to matrices over $\GF(2^p)$. 
The following is the outline of construction procedure. 
\begin{enumerate}[A{)}]
\item
Choose integers $J=2,P$ and $L$ such that $PL=N$. 
Construct  a pair of orthogonal
 $PJ\times PL$ quasi-cyclic binary matrices $\hHC$ and $\hHD$ 
 by following the procedure of \cite{quant-ph/0701020v4}.
\item
A pair of  $PJ\times PL$ matrices $ \HG$ and $\HD$ over $\GF(2^p)$ are constructed 
by replacing non-zero entries of the matrices $\hHC$ and $\hHD$ with  elements of $\GF(2^p)$. 
The orthogonality condition $ \HG\HD^{\mathsf{T}}=0$ imposes a set of linear constraints on the non-zero entries of 
$ \HG$ and $\HD$, that can be solved by Gaussian elimination. 
%A counting argument also shows that these equations admit non-trivial solutions.
\item
 The entries of $ \HG$ and $\HD$ are mapped to elements of $\GF(2)^{p \times p}$ that
 preserves addition and multiplication. These properties of the mapping
 imply that the resulting matrices $H_C$ and $H_D$ remain orthogonal.
\end{enumerate}
The proposed code can be viewed as a concatenated code with an inner LDPC code and an outer code of rate 1 and length $p$. 
Indeed, the code obtained at step A) is a quantum LDPC code \cite{quant-ph/0701020v4}. 
Steps B) and C) have the effect of replacing each qubit by $p$ qubits while preserving the orthogonality property. 

Each step in the procedure is explained in the following sub-sections. 
%---------------------- QC LDPC ---------------------------
\subsection{Quasi-Cyclic Binary Construction}
%---------------------- QC LDPC ---------------------------
Let $\hHC$ and $\hHD$ be $PJ\times PL$ binary parity-check matrices defined as follows:
 \begin{align*}
 &\hHC := (I(c_{j, \ell}))_{0\le j<J, 0\le \ell<L}, \\
 &\hHD := (I(d_{j, \ell}))_{0\le j<J, 0\le \ell<L}, \\
 &I(1) :=
  \begin{bmatrix}
   0 & 1 & 0 & 0 & 0 \\
   0 & 0 & 1 & 0 & 0 \\
   0 & 0 & 0 & \ddots & 0 \\
   0 & 0 & 0 & 0 & 1 \\
   1 & 0 & 0 & 0 & 0
  \end{bmatrix}\in \{0, 1\}^{P\times P}, \\
 &I(c_{j, \ell}) :=I(1)^{c_{j, \ell}}.
\end{align*}
We refer to such matrices as ($J$, $L$, $P$)-QC matrices. 

%------------------------Hagiwara Imai-------------------------------------
Hagiwara and Imai proposed \cite{4557323} the following method for constructing a QC parity-check matrix pair $(\hHC, \hHD)$. 
In the original paper \cite{4557323}, the construction method is more flexible about the row size of the matrices, i.e., $\hHC$ and $\hHD$ can have different row sizes. 
For simplicity, in this paper,  we focus on $\hHC$ and $\hHD$ with the same row size $JP$. 
\begin{theorem}[{\cite[Theorem 5.2]{quant-ph/0701020v4}}]
\label{Hagiwara-Imai}
Define 
 $\mathbb{Z}_{P}^{*} := \{ z \in \mathbb{Z}_{P}\mid \exists a \in \mathbb{Z}_{P}, za=1 \}, $ and
 $\mathrm{ord}(\sigma) := \min\{ m > 0 \mid \sigma^m = 1\}.$
For  integers $P>2, J, L, 0\le \sigma<P$ and $0\le\tau <P$ such that 
\begin{align}
&\sigma, \tau\in\mathbb{Z}_{P}^{*},\label{164916_7Jul10}\\
& L/2 =  \mathrm{ord}(\sigma), \label{231740_4Jul10}\\
& 1 \le J \le \mathrm{ord}(\sigma), \nonumber\\
&\mathrm{ord}(\sigma) \neq \# \mathbb{Z}_{P}^{*}, \nonumber\label{211218_6Jul10}\\  
&1- \sigma^{j} \in \mathbb{Z}_{P}^{*} \text{ for all } 1 \le j < \mathrm{ord}(\sigma), \\
&\tau \neq 1, \sigma, \sigma^2, \dots, \sigma^{\mathrm{ord}(\sigma)-1},\label{215258_7Jul10} 
\end{align}
let $\hHC$ and $\hHD$ be two $(J, L, P)$-QC binary matrices such that 
\begin{align}
&\hHC=(I(c_{j, \ell}))_{0\le j<J, 0\le\ell<L}, \nonumber\\
&\hHD=(I(d_{j, \ell}))_{0\le j<J, 0\le\ell<L}, \nonumber\\
&   c_{j, \ell} := 
 \left\{
  \begin{array}{rc}
        \sigma^{-j + \ell}     &  0 \le \ell < L/2\label{051947_15Jun10}  \\
   \tau \sigma^{-j + \ell}     &  L/2 \le \ell < L,  \\
  \end{array}
 \right.\\
 & d_{j, \ell} := 
 \left\{
  \begin{array}{rc}
    -\tau \sigma^{j-\ell}     &  0 \le \ell < L/2\label{051527_15Jun10}  \\
        - \sigma^{j-\ell}     &  L/2 \le \ell < L  \\
  \end{array}
 \right.
\end{align} 
then
it holds that $\hat{H}_{C}\hHD^{\mathsf{T}} = 0$ and there are no cycles of size 4 
in the Tanner graph of $\hHC$ and $\hHD$. 
\end{theorem}
%-------------------  Example Figure ---------------------------
\begin{figure*}[t]
\begin{center}
 \begin{tabular}{cc}
 \begin{minipage}{0.47\textwidth}
\small{
\renewcommand\arraystretch{0.5}
\begin{tabular}
{|@{}l@{}@{}l@{}@{}l@{}@{}l@{}@{}l@{}@{}l@{}@{}l@{}|@{}l@{}@{}l@{}@{}l@{}@{}l@{}@{}l@{}@{}l@{}@{}l@{}|@{}l@{}@{}l@{}@{}l@{}@{}l@{}@{}l@{}@{}l@{}@{}l@{}||@{}l@{}@{}l@{}@{}l@{}@{}l@{}@{}l@{}@{}l@{}@{}l@{}|@{}l@{}@{}l@{}@{}l@{}@{}l@{}@{}l@{}@{}l@{}@{}l@{}|@{}l@{}@{}l@{}@{}l@{}@{}l@{}@{}l@{}@{}l@{}@{}l@{}|@{}}
\hline
&\tt{4}&&&&&&&&\tt{9}&&&&&&&&&\tt{a}&&&&&&\tt{2}&&&&&&&&&&\tt{b}&&&&&&\tt{8}&\\
&&\sf{5}&&&&&&&&\tt{6}&&&&&&&&&\tt{2}&&&&&&\sf{b}&&&\tt{5}&&&&&&&&&&&&&\tt{d}\\
&&&\tt{d}&&&&&&&&\tt{c}&&&&&&&&&\sf{9}&&&&&&\tt{d}&&&\sf{c}&&&&&&\tt{0}&&&&&&\\
&&&&\tt{2}&&&&&&&&\tt{8}&&\tt{c}&&&&&&&&&&&&&\tt{c}&&&\tt{2}&&&&&&\tt{9}&&&&&\\
&&&&&\tt{a}&&&&&&&&\tt{7}&&\tt{6}&&&&&&\tt{b}&&&&&&&&&&\tt{3}&&&&&&\tt{2}&&&&\\
&&&&&&\tt{d}&\sf{9}&&&&&&&&&\tt{6}&&&&&&\tt{c}&&&&&&&&&&\tt{5}&&&&&&\sf{7}&&&\\
\tt{e}&&&&&&&&\tt{0}&&&&&&&&&\tt{c}&&&&&&\tt{e}&&&&&&&&&&\tt{3}&&&&&&\tt{1}&&\\\hline
&&&&\tt{9}&&&&\tt{2}&&&&&&&&\tt{1}&&&&&&&&&&\tt{5}&&&&&\tt{4}&&&&&&&&&&\tt{a}\\
&&&&&\tt{a}&&&&\tt{4}&&&&&&&&\tt{d}&&&&&&&&&&\tt{5}&&&&&\tt{0}&&&\tt{d}&&&&&&\\
&&&&&&\tt{1}&&&&\tt{8}&&&&&&&&\tt{4}&&&\tt{1}&&&&&&&&&&&&\tt{d}&&&\tt{0}&&&&&\\
\tt{e}&&&&&&&&&&&\tt{5}&&&&&&&&\tt{a}&&&\tt{1}&&&&&&&&&&&&\tt{b}&&&\tt{6}&&&&\\
&\tt{d}&&&&&&&&&&&\tt{d}&&&&&&&&\sf{e}&&&\tt{7}&&&&&\tt{5}&&&&&&&&&&\sf{c}&&&\\
&&\sf{e}&&&&&&&&&&&\tt{e}&\tt{5}&&&&&&&&&&\tt{7}&&&&&\sf{5}&&&&&&&&&&\tt{8}&&\\
&&&\tt{a}&&&&\sf{3}&&&&&&&&\tt{d}&&&&&&&&&&\sf{6}&&&&&\tt{a}&&&&&&&&&&\tt{8}&\\\hline
\end{tabular}

} \end{minipage}
 & \begin{minipage}{0.47\textwidth}
 \small{ 
\renewcommand\arraystretch{0.538}
\begin{tabular}
{|@{}l@{}@{}l@{}@{}l@{}@{}l@{}@{}l@{}@{}l@{}@{}l@{}|@{}l@{}@{}l@{}@{}l@{}@{}l@{}@{}l@{}@{}l@{}@{}l@{}|@{}l@{}@{}l@{}@{}l@{}@{}l@{}@{}l@{}@{}l@{}@{}l@{}||@{}l@{}@{}l@{}@{}l@{}@{}l@{}@{}l@{}@{}l@{}@{}l@{}|@{}l@{}@{}l@{}@{}l@{}@{}l@{}@{}l@{}@{}l@{}@{}l@{}|@{}l@{}@{}l@{}@{}l@{}@{}l@{}@{}l@{}@{}l@{}@{}l@{}|@{}}
\hline
&&&&\tt{2}&&&&&\tt{8}&&&&&&\tt{4}&&&&&&&&&&&&\tt{7}&&&&\tt{7}&&&&&&&&&\tt{9}&\\
&&&&&\tt{7}&&&&&\tt{e}&&&&&&\tt{1}&&&&&\tt{6}&&&&&&&&&&&\tt{2}&&&&&&&&&\tt{7}\\
&&&&&&\tt{7}&&&&&\tt{4}&&&&&&\tt{1}&&&&&\tt{8}&&&&&&&&&&&\tt{a}&&\tt{1}&&&&&&\\
\tt{2}&&&&&&&&&&&&\tt{b}&&&&&&\tt{6}&&&&&\tt{2}&&&&&&&&&&&\tt{5}&&\tt{a}&&&&&\\
&\tt{8}&&&&&&&&&&&&\tt{3}&&&&&&\tt{4}&&&&&\tt{a}&&&&\tt{1}&&&&&&&&&\tt{8}&&&&\\
&&\sf{d}&&&&&\sf{a}&&&&&&&&&&&&&\sf{a}&&&&&\sf{7}&&&&\sf{7}&&&&&&&&&\sf{c}&&&\\
&&&\tt{3}&&&&&\tt{6}&&&&&&\tt{8}&&&&&&&&&&&&\tt{3}&&&&\tt{3}&&&&&&&&&\tt{5}&&\\\hline
&\tt{b}&&&&&&&&&&\tt{a}&&&&&\tt{d}&&&&&&&&&&\tt{9}&&&&&&&&\tt{4}&&&&\tt{c}&&&\\
&&\tt{a}&&&&&&&&&&\tt{2}&&&&&\tt{5}&&&&&&&&&&\tt{d}&\tt{a}&&&&&&&&&&&\tt{1}&&\\
&&&\tt{9}&&&&&&&&&&\tt{1}&&&&&\tt{9}&&&\tt{c}&&&&&&&&\tt{a}&&&&&&&&&&&\tt{b}&\\
&&&&\tt{0}&&&\tt{7}&&&&&&&&&&&&\tt{a}&&&\tt{4}&&&&&&&&\tt{0}&&&&&&&&&&&\tt{e}\\
&&&&&\tt{0}&&&\tt{9}&&&&&&&&&&&&\tt{3}&&&\tt{a}&&&&&&&&\tt{7}&&&&\tt{c}&&&&&&\\
&&&&&&\tt{d}&&&\tt{2}&&&&&\tt{b}&&&&&&&&&&\tt{9}&&&&&&&&\tt{6}&&&&\tt{e}&&&&&\\
\tt{6}&&&&&&&&&&\tt{7}&&&&&\tt{a}&&&&&&&&&&\tt{2}&&&&&&&&\tt{2}&&&&\tt{e}&&&&\\\hline
\end{tabular}
 }\end{minipage}
 \end{tabular}
\end{center}
\caption{An example of non-binary matrices $\HG=(\gamma_{m, n})_{0\le m<M, 0\le n<N}$ and 
$\HD=(\delta_{m, n})_{0\le m<M, 0\le n<N}$ over $\GF(2^4)$ such that $\HG\HD^{\mathsf{T}}=0$ with $M=14$ and $N=42$. 
Each non-zero entry is represented as the hexadecimal number of $\log_\alpha(\gamma_{m, n})$, where $\alpha$ is a primitive element such that 
$\alpha^4+\alpha+1=0$.
E.g., $\alpha^{0}$ and $\alpha^{11}$ are represented as 
 0 and b,  respectively. }
\label{130803_26Jun10}
\end{figure*}
%-------------------------------------------------------------
From Theorem \ref{Hagiwara-Imai}, we obtain two $JP\times LP$ binary matrices $\hHC$ and $\hHD$
such that $\hHC\hHD^{\mathsf{T}}=0$ and the Tanner graphs of $\hHC$ and $\hHD$ are free of cycles of size 4. 
 We give an example. 
\begin{example}
\label{ex1}
 With parameters $J=2, L=6, P=7, \sigma=2$ and $\tau=3$, 
 from Theorem \ref{Hagiwara-Imai}, we are given a $JP\times LP$ binary matrix pair $(\hHC, \hHD)$ such that $\hHC\hHD^{\mathsf{T}}=0$ as follows. 
 \begin{align*}
 & \hHC=
 \begin{pmatrix}
 I(1)&I(2)&I(4)&I(3)&I(6)&I(5)\\
 I(4)&I(1)&I(2)&I(5)&I(3)&I(6)
 \end{pmatrix}, 
\\
& \hHD=
 \begin{pmatrix}
 I(4)&I(2)&I(1)&I(6)&I(3)&I(5)\\
 I(1)&I(4)&I(2)&I(5)&I(6)&I(3)
 \end{pmatrix}.
 \end{align*}
The binary representation of these matrices are given in Fig.~\ref{134801_6Jul10}. 
It can be verified that there are no cycles of length 4 in Tanner graphs of $\hHC$ and $\hHD$. 
\end{example}

% --------------- Observation  ----------------------
We observe a fundamental property of $\hHC$ and $\hHD$ as follows.
\begin{observation}
Let $m':=5$.  The $m'$-th row of $\hHD$ has non-zero entries at the $n_0=2, n_2=7, n_4=20, n_1=25, n_5=29$ and  $n_3=38$-th columns. 
Denote the set of these indices by $N(m'):=\{n_0,\dotsc,n_5\}$. 
Note that the index starts from 0. These non-zero entries are represented in thick font in Fig.~\ref{134801_6Jul10}.
At each of these 6 columns in $N(m')$, $\hHC$ has $J=2$ non-zero entries at
\begin{align*}
\begin{array}{r@{}@{}r@{}r@{}r@{}r@{}r@{}r@{}r@{}r@{}r@{}r@{}r@{}r@{}r@{}r@{}r@{}r@{}}
\setlength{\arraycolsep}{0pt}
 (& 1, &2),~ &(& 5, &7),~ &( 2, &20),~ &(&1,  &25),~  &(&2,  &29),~&(&5,  &38), \\
 (&12, &2),~ &(&13, &7),~ &(11, &20),~ &(&13, &25),~  &(&12, &29),~&(&11, &38). 
\end{array}
\end{align*}
Let these positions be denoted, respectively, by 
\begin{align*}
\begin{array}{@{}l@{~}l@{~}l@{~}l@{~}l@{~}l@{ }}
  (m_0,  n_0), &(m_2,  n_2), &(m_4,  n_4), &(m_0, n_1),  & (m_4,  n_5),&(m_2,  n_3), \\
 (m_5, n_0), & (m_1, n_2),& (m_3, n_4), &(m_1, n_1),&  (m_5, n_5),&   (m_3, n_3). 
\end{array}
\end{align*}
We denote the set of these positions by $E(m')$. 
In the Tanner graph of $\hHC$, those non-zero entries from  $(m_0,n_0)$ to $(m_5,n_0)$ consist of a cycle of size $2L=12$. 
Along this cycle, the row index  moves back and forth between the above and below of $\hHC$. Indeed,  $0\le m_{2 i}< P\le m_{2 i+1}<2P$ for $i=0,\dotsc,L/2-1$. 
On the other hand, the column index moves back and forth between the left and right of $\hHC$. Indeed,  $0\le n_{2 i}<LP/2\le n_{2 i+1}<LP$ for $ i=0,\dotsc,L/2-1$.
\end{observation}

We claim that this observation is general for any $m'$-th row of $\hat{D}_C$. 
The following lemma will be a key ingredient for constructing non-binary matrices $\HD$ and $\HG$ in Section \ref{subsection:nb}.
% --------------- Lemma ---------------------------
\begin{lemma}\label{lemma:II.1}
Let $\hHC, \hHD$ be 
the two (2, $L$, $P$)-QC binary matrices dealt in Theorem \ref{Hagiwara-Imai}.
Let $N{(m')} := \{n_{0}, \dotsc, n_{L-1} \}$ be the support of the $m'$-th row of $H_D$. 
To be precise, 
\begin{align*}
 N{(m')} = \{n_{0}, \dotsc, n_{L-1}\}=\{0\le n<LP\mid \hat{d}_{m', n}\neq 0\}.
\end{align*}
Let $E{(m')}$ be the set of non-zero entry positions in $\hHC$ whose column is in $N{(m')} $. 
To be precise, 
\begin{align}
 E{(m')}:=\{(m, n)\mid \hat{c}_{m, n}\neq 0, n\in N{(m')}\}.
\end{align}
In this setting, in the Tanner graph of $\hHC$, for any $m'=0, \dotsc, L-1$, 
the $L$ variable nodes corresponding to the column index in $N{(m')}$ 
and the $L$ adjacent check nodes form a cycle of length $2L$.
In other words, 
there exist $L$ distinct $m_0, \dotsc, m_{L-1}$ and $L$ distinct $n_0, \dotsc, n_{L-1}$, such that 
\begin{align}
\begin{split}
 &  \{(m_{2 i-1}, n_{2 i}), (m_{2 i}, n_{2 i} ), (m_{2 i}, n_{2 i+1} ), (m_{2 i+1}, n_{2 i+1} )\\
 &\mid 0\le i<L/2\}= E{(m')},\label{104307_12Jun10}
\end{split}
\end{align}
where we denote $m_{-1} := m_{L-1}$ and
$m_{L+1}:= m_{1}$. 
% In other words, 
% the following is a walk on the Tanner graph
% \begin{align}
% &  (m_0, n_0), (m_0, n_1) , (m_1, n_1), (m_1, n_2), \dotsc, (m_{L-1}, n_{L-1}), \allowbreak (m_{L-1}, n_{0}). \label{104307_12Jun10}
% \end{align}
\end{lemma}
%-------------------proof ---------------------------------
{\itshape Proof}:
Since $J=2$, it follows that $\hHC$ and $\hHD$ comprise $2\times L$ sub-matrices of size $P\times P$. 
For simplicity, we focus on the $m'$-th row chosen from the upper half rows of $\hHD$, i.e., $0\le m'<P$. 
The proof for $P\le m'<2P$ is essentially the same. 

First, we will clarify the support $N(m')$ of the $m'$-th row of $\hHD$. 
From \eqref{051527_15Jun10}, the upper half of $\hHD$ can be written by $L$ sub-matrices as follows. 
\begin{align*}
 \begin{pmatrix}
I(-\tau \sigma^{-0})\cdots I(-\tau \sigma^{-(L/2-1)})||  I(- \sigma^{-L/2} )\cdots I(-\sigma^{-(L-1)})
 \end{pmatrix}. 
\end{align*}
For each sub-matrix, the $m'$-th row has the only one non-zero entry. 
For $0\le i<L/2$, let $n_{2 i}$ be the column index of such a non-zero entry in the $ i$-th sub-matrix. 
Similarly, let $n_{2 i+1}$ be the column index of such a non-zero entry in the $([- i]_{L/2}+L/2)$-th sub-matrix. 
Obviously, $N{(m')} = \{n_{0}, \dotsc, n_{L-1}\}$ and $n_0,\dotsc,n_{L-1}$ are distinct.
%Note that the $n_{0}$-th, $n_{2}$-th, $\dotsc$, $n_{L-2}$-th columns are in the $L/2$ left half  $0$-th, 1-st, $\dotsc$, $(L/2-1)$-th  sub-matrices of $\hHD$, respectively, 
Note that the $n_{3}$-th, $\dotsc$,  $n_{L-3}$-th and $n_1$-th columns are in the $L/2$ right half $(L-1)$-th, $\dotsc,(L/2+1)$-th and $(L/2)$-th sub-matrices of $\hHD$, respectively. 
The $m'$-th row of $I(x)$ has the only non-zero entry at the $(x+m')$-th column. 
From this,  it follows that we can rewrite 
\begin{align}
   n_{2 i } &= [ -\tau \sigma^{-  i } + m' ]_P +   i P,\label{224828_14Mar11}\\
 n_{2 i+1 } &= [ - \sigma^{[- i]_{L/2}} + m' ]_P + ([- i]_{L/2}+L/2) P, 
\end{align}
where we define $[x]_t\in \mathbb{Z}$ for $ x \in \mathbb{Z}$ as $0\le [x]_t <t$ such that  $[x]_t = x \ \pmod{t}$. 

% To be precise, 
% \begin{align*}
%   i P\le &n_{2 i}<( i+1)P, \\
%  ([- i]_{L/2}+L/2)P\le &n_{2 i+1}<([- i]_{L/2}+L/2+1)P. 
% \end{align*}

\begin{figure*}[t]
\begin{center}
 \begin{tabular}{c|c}
  \includegraphics[width=0.45\textwidth, height=0.25\textwidth]{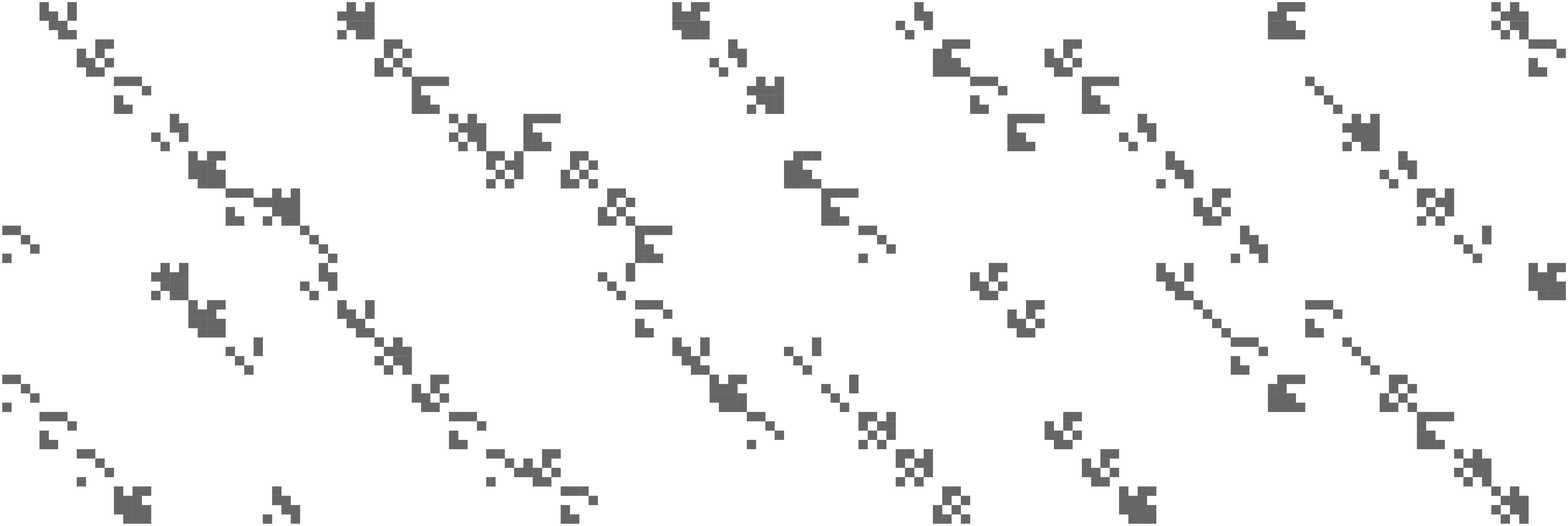} 
 & \includegraphics[width=0.45\textwidth, height=0.25\textwidth]{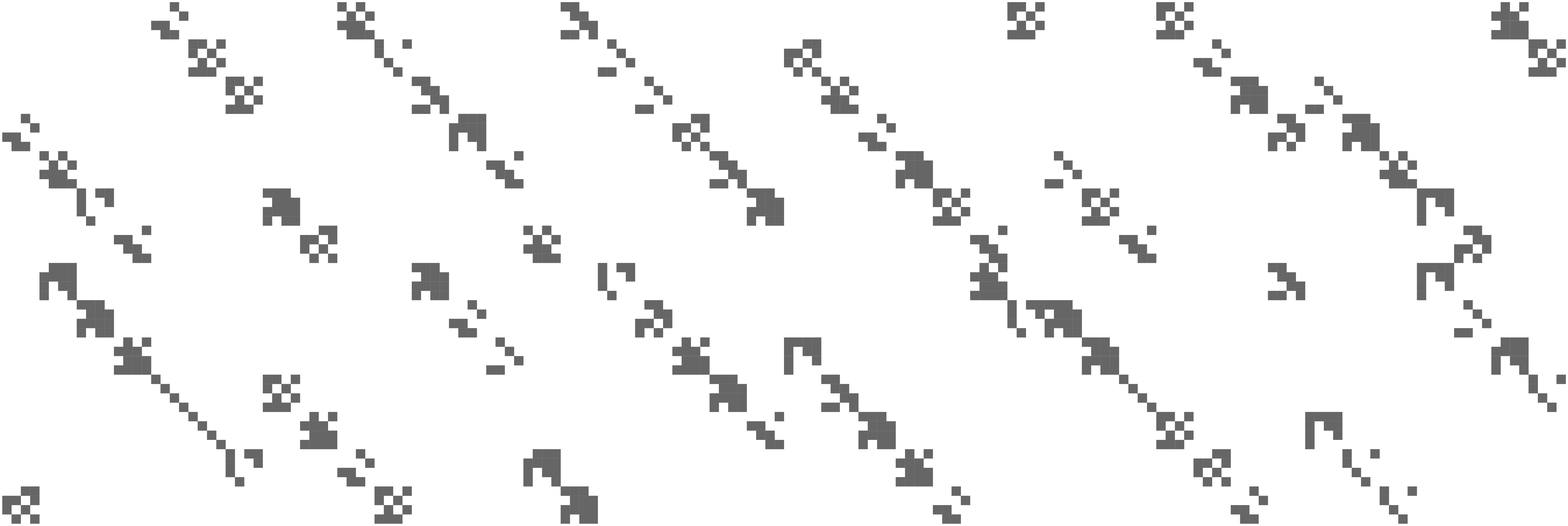} 
 \end{tabular} 
\end{center}
 \caption{An example of binary $pM\times pN$ matrices $H_C$ and $H_D$ such that $H_CH_D^{\mathsf{T}}=0$ with $p=4, M=14$ and $N=42$. 
Non-zero entries are represented in black.
The codes have many cycles of size 4 as binary codes. On the other hand, the codes have no cycles of size 4 as non-binary codes.}
 \label{113821_28Jun10}
\end{figure*}
%------------------------------
Secondly, we will prove \eqref{104307_12Jun10} by clarifying the structure of $E(m')$. 
Define
\begin{align*}
 m_{2 i}:&=[ - \sigma^{ i }  - \tau \sigma^{- i} + m']_P,\\ 
 m_{2 i-1}:&=[- \sigma^{ i-1 }  - \tau \sigma^{- i}+ m']_P + P. 
\end{align*}
We claim that $(m_{2 i-1}, n_{2 i}), (m_{2 i}, n_{2 i} ) \in E{(m')}$ for $0\le i<L/2$. 
Consider the non-zero entry in the $(j,  i)$-th sub-matrix for $j=0,1,$ and $0\le i<L/2$.
Let  the position be denoted by  $(m,n_{2  i})$. 
%Let $m$ be rewritten by $m=[m]_{P}+jP$. 
From \eqref{051947_15Jun10} and  $ i<L/2$, it follows that 
the $m$-th row  of $\hHC$ has non-zero entries at the $([\sigma^{-j+ i}+m]_P + i P)$-th column. 
Therefore, it follows from \eqref{224828_14Mar11} that $(m, n_{2 i} ) \in E{(m')}$ if and only if 
\begin{align*}
   [\sigma^{-j+ i}+m]_P + i P&=[ - \tau\sigma^{- i} + m' ]_P +  i P, 
\end{align*}
for $j=0, 1$. 
Thus, we conclude that $m=m_{2 i-1}$  and $m=m_{2 i}$ for $j=0,1$, respectively, which proves the claim.
In a similar manner, it can be shown that  
$(m_{2 i}, n_{2 i+1} ),(m_{2 i+1}, n_{2 i+1} ) \in E{(m')}$ for $0\le i<L/2$, where we denoted $m_{L+1}:= m_{1}$. 
These prove \eqref{104307_12Jun10}.

Finally, we show that the $2L$ non-zero entries in $E{(m')}$  consist of a cycle of length $2L$.
To this end, it is suffice to show that 
$m_0,m_2,\dotsc,m_{L-2}$ are distinct, and 
$m_1,m_3,\dotsc,m_{L-1}$ are distinct. Assume $m_{2 i}=m_{2 i'}$ for $ i\neq i'$. Then we have 
$\sigma^{ i }  + \tau \sigma^{- i} +m'= \sigma^{ i' }  +\tau \sigma^{- i'}+m' \pmod P$. 
Some calculations reveal that $(1 - \sigma^{ i'- i})(\sigma^{ i } - \tau \sigma^{- i'})=0 \pmod P.$
From \eqref{211218_6Jul10} it follows that $\sigma^{ i } - \tau \sigma^{- i'}=0 \pmod P$, which contradicts with \eqref{215258_7Jul10}. Hence, we conclude $ i= i'$.
In the same manner, we can show $m_1,m_3,\dotsc,m_{L-1}$ are distinct.
%------------------------------
\qed
% --------------- NB H construction  ----------------------
\subsection{Non-binary matrix pair construction}
\label{subsection:nb}
Define $M:=JP$ and $N:=LP$. 
So far, we obtain orthogonal matrices 
$\hHC=(\hat{c}_{m, n})_{0\le m < M, 0\le n < N}$ and $\hHD=(\hat{d}_{m, n})_{0\le m < M, 0\le n< N}$ whose Tanner graphs are free of cycles of size 4. 
In this sub-section, we will construct orthogonal non-binary $M\times N$ matrices $ \HG=(\gamma_{m, n})_{0\le m<M, 0\le n<N}$ and $\HD=(\delta_{m, n})_{0\le m<M, 0\le n<N}$ over $\GF(2^p)$  such that $\gamma_{m, n}\neq 0$ iff $\hat{c}_{m, n}\neq 0$ and $\delta_{m, n}\neq 0$ iff $\hat{d}_{m, n}\neq 0$.
Obviously, the Tanner graphs of $H_{\Gamma}$ and $H_{\Delta}$ are free of cycles of size 4. 
We will determine the non-zero entries of $H_{\Gamma}$ and $H_{\Delta}$ such that $H_{\Gamma} H_{\Delta}^{\mathsf{T}} = 0$ in the rest of this sub-section.

% -------------- H NB construction cont.
For $\HG \HD^{\mathsf{T}}=0$, it is required that  the $m'$-th row of $\HD$ is in the null-space of $\HG$ for each $0\le m'<JP$. 
From Lemma \ref{lemma:II.1}, this is equivalent to 
\begin{align}
& 
\begin{bmatrix}
  \gamma_{m_0, n_0}       & \hspace{-2mm}\gamma_{m_0, n_1} & \hspace{-2mm}                          &\hspace{-2mm}\label{100120_12Jun10}                                    \\ 
                          & \hspace{-2mm}\ddots            & \hspace{-2mm}\ddots                    &\hspace{-2mm}                                    \\
                          & \hspace{-2mm}                  & \hspace{-2mm}\gamma_{m_{L-2}, n_{L-2}} &\hspace{-2mm} \gamma_{m_{L-2}, n_{L-1}} \\
 \gamma_{m_{L-1}, n_{0}} & \hspace{-2mm}                  & \hspace{-2mm}                          &\hspace{-2mm} \gamma_{m_{L-1}, n_{L-1}}
\end{bmatrix}
\begin{bmatrix}
 \delta_{m', n_0}\\\vdots\\\vdots\\\delta_{m', n_{L-1}}
\end{bmatrix}=0,
% &  \gamma_{m_0, n_0} \delta_{m', n_0}+ \gamma_{m_0, n_1}\delta_{m', n_1} =0\\
% &  \gamma_{m_1, n_1} \delta_{m', n_0}+ \gamma_{m_1, n_2}\delta_{m', n_2} =0\\
% & \qquad \qquad \qquad \vdots\\
% &  \gamma_{m_{L-1}, n_{L-1}} \delta_{m', n_{L-1}}+ \gamma_{m_{L-1}, n_{0}}\delta_{m', n_0} =0
\end{align}
for each $0\le m'<JP$. 
In order to find the non-zero entries of $\HG$ and $\HD$, 
this equation needs to have non-trivial solutions, i.e., the determinant of the left matrix, denoted by $\Gamma_{m'}$, in \eqref{100120_12Jun10} is 0:
\begin{align}
 \mathrm{det}( \Gamma_{m'} ) = \gamma_{m_0, n_0}\cdots\gamma_{m_{L-1}, n_{L-1}} - \gamma_{m_0, n_1}\cdots\gamma_{m_{L-1}, n_{0}}=0.\label{104959_12Jun10}
\end{align}
Divide $E{(m')}$ in \eqref{104307_12Jun10} into two parts as in the proof of Lemma \ref{lemma:II.1}: 
\begin{align*}
& E{(m')}=E_1{(m')}\cup E_2{(m')}, \\
&E_1{(m')}:= \{(m_0, n_0), (m_1, n_1), \dotsc, (m_{L-1}, n_{L-1})\}, \\
&E_2{(m')}:= \{(m_0, n_1), (m_1, n_2), \dotsc, (m_{L-1}, n_{0})\}.
\end{align*}
Then \eqref{104959_12Jun10} can be transformed to
\begin{align}
 \prod_{(m, n)\in E_1{(m')}}\gamma_{m, n}\prod_{(m, n)\in E_2{(m')}}\gamma_{m, n}^{-1}=1.
\end{align}
For $\alpha^x\in \GF(2^p)$, define $\log_\alpha(\alpha^x):=x \pmod {2^{p}-1 } $. Then $\log_\alpha$ is well-defined. The equation above is equivalent to the following linear equation  over $\Z_{2^p-1}$. 
\begin{align}
 \sum_{(m, n)\in E_1{(m')}}\log_\alpha\gamma_{m, n}-\sum_{(m, n)\in E_2{(m')}}\log_\alpha\gamma_{m, n}=0.\label{032326_28Jun10}
\end{align}
Thus, we have $JP$ linear equations over $\Z_{2^p -1}$ for $m'=0, \dotsc, JP-1$. 
Solving these linear equations by the Gaussian elimination, we get the candidate solution space of the non-zero entries of $\HG$ such that \eqref{032326_28Jun10} holds for  $m'=0, \dotsc, JP-1$.
Picking non-zero entries of $\HG$ randomly from the candidate solution space and solving \eqref{100120_12Jun10}, we obtain non-zero entries of $\HD$.
We give an example. 
\begin{example}
\label{ex2}
Using $\hHC$ and $\hHD$ given in Example \ref{ex1}, we get an
$M\times N$ non-binary matrix pair $(\HG, \HD)$ over $\GF(2^p)$ such that $\HG\HD^{\mathsf{T}}=0$ with $M=JP=14$ and $N=LP=42$.
The resulting $(\HG, \HD)$ is depicted in Fig.~\ref{130803_26Jun10}.
\end{example}
This construction can be viewed as picking $(H_{\Gamma}, H_{\Delta})$ randomly from $\{(H_{\Gamma}, H_{\Delta})\mid H_{\Gamma}H_{\Delta}^{\mathsf{T}}=0\}$, where  $H_{\Gamma}$ and $H_{\Delta}$ are constrained 
to have non-zero entries at the same positions as $\hHC$ and $\hHD$, respectively. 
Since $\hHC$ and $\hHD$ is equivalent with some column permutation \cite{4557323}, the construction has symmetry for $H_{\Gamma}$ and $H_{\Delta}$. 
This symmetry leads to almost the same decoding performance which will be observed by computer experiments in Section \ref{231523_10Jun10}. 
%%%%%%%%%%%%%%%%%%%%%%%%%%%%%%%%%%%%%%%%%%
\subsection{Binary Quasi-Cyclic CSS LDPC Codes}
\label{223910_10Jun10}
%%%%%%%%%%%%%%%%%%%%%%%%%%%%%%%%%%%%%%%%%%
So far, we obtained $M\times N$ sparse non-binary $\GF(2^p)$ parity-check matrices $\HD$ and $\HG$, where  $N:=PL$ and $M:=PJ$. 
It is known that non-binary codes have the binary representation of their parity-check matrices. 
In this section, we show that two parity-check matrices $\HG$ and $\HD$ over $\GF(2^p)$ such that $\HG \HD^{\mathsf{T}}=0$ can be represented by two binary matrices 
$H_C$ and $H_D$ such that $H_CH_D^{\mathsf{T}}=0$. 

Let $\GF(2^p)$ has a primitive element $\alpha$ with its primitive polynomial $\pi(x)=\sum_{i=0}^{p-1}\pi_ix^i+x^p$. 
It is known \cite{macwilliams77} that the following map $A$ from $\GF(2^p)$ to $\GF(2)^{p\times p}$ is bijective and its image is isomorphic to $\GF(2^p )$ as a field by sum and multiple as matrices. 
  \begin{align*}
 & \GF(2^p)\ni \alpha^i \mapsto A({\alpha^i}):=A(\alpha)^i\in \GF(2)^{p\times p}, \\
 &A(0)=0, \\
&A(\alpha):=\begin{bmatrix}
      0 & 0 & 0 & 0 & \pi_0\\
      1 & 0 & 0 & 0 & \pi_1\\
      0 & 1 & 0 & 0 & \pi_2\\
      \vdots & \vdots & \ddots & \vdots\\
      0 & 0 & 0 & 1 & \pi_{p-1}
	    \end{bmatrix}.
 \end{align*}
Moreover, it holds that 
 \begin{align*}
 &A({\alpha^i})\underline{v}(\alpha^j)=\underline{v}(\alpha^{i+j}), \\
 & \text{where } \alpha^i=\textstyle\sum_{j=0}^{p-1}a_j\alpha^j\in\GF(2^p), \\ 
 & \text{and } \underline{v}(\alpha^i) := (a_0, \dotsc, a_{m-1})^\mathsf{T}\in \GF(2)^p.\\
 \end{align*}
Furthermore, with an abuse of notation we define $A(\underline{v}(\alpha^j)):=\underline{v}( \alpha^j )$. 
%-------------------  Fact  ---------------------------
\begin{fact}
 Let $\HG$ and $\HD$ be matrices over $\GF(2^{p})^{M \times N}$ and 
 let ${H}_C$ and ${H}_D$ be two matrices over $\GF( 2 )^{pM \times pN }$ such that 
 \begin{align*}
 & \HG=(\gamma_{m, n})_{0\le m<M, 0\le n<N}, \\
 &  \HD=(\delta_{m, n})_{0\le m<M, 0\le n<N}, \\
 & {{H}}_C=(A(\gamma_{m, n}))_{0\le m<M, 0\le n<N}, \\
 &  {{H}}_D=(A^{\mathsf{T}}(\delta_{m, n}))_{0\le m<M, 0\le n<N}.
 \end{align*}
Then, it holds that if $\HG \HD^\mathsf{T}=0$, then ${H}_C{H}_D^\mathsf{T}=0$. 
\end{fact}
\begin{proof}
 Let $({H}_C{H}_D^\mathsf{T})_{m, n}$ be the $(m, n)$-th $p\times p$ binary sub-matrix of ${H}_C{H}_D^\mathsf{T}$, and
let $(\HG \HD^\mathsf{T})_{m, n}$ be the $(m, n)$-th entry of $\HG \HD^\mathsf{T}$. 
 Then, for any $0\le m<M$ and $0\le n<N$, 
 \begin{align*}
  ({H}_C{H}_D^\mathsf{T})_{m, n}&=\textstyle\sum_{k=0}^{N-1}A(\gamma_{m, k})A(\delta_{k, n})\\
 &=\textstyle\sum_{k=0}^{N-1}A(\gamma_{m, k}\delta_{n, k})\\
 &=\textstyle A(\sum_{k=0}^{N-1}\gamma_{m, k}\delta_{n, k})\\
 &=A((\HG \HD^\mathsf{T})_{m, n})=A(0)=0. 
 \end{align*}
\end{proof}
%-------------------  Example  ---------------------------
\begin{example}
Using $\HG$ and $\HD$ given in Example \ref{ex2}, we get a
$pM\times pN$ binary matrix pair $({H}_C, {H}_D)$
 such that ${H}_CH_D^{\mathsf{T}}=0$ with $p=4, pM=pJP=56$ and $pN=pLP=168$.
The resulting $({H}_C, {H}_D)$ is depicted in Fig.~\ref{113821_28Jun10}.
\end{example}
%-----------------------------------
%%%%%%%%%%%%%%%%%%%%%%%%%%%%%%%%%%%%%%%%%%%%%%%%%
\section{Decoding Algorithm}
\label{231512_10Jun10}
%%%%%%%%%%%%%%%%%%%%%%%%%%%%%%%%%%%%%%%%%%%%%%%%%
In this section, we describe the decoding algorithm for the CSS code pair $(C, D)$ constructed by the proposed method in Section \ref{230852_10Jun10} and \ref{223910_10Jun10}. 
The decoding algorithm is based on the decoding algorithm of classical non-binary LDPC codes \cite{706440}. 
The input of the decoding algorithm is the syndrome. 
We assume the depolarizing channels \cite[Section V]{1337106} with depolarizing probability $2f_\mathrm{m}/3$, where $f_\mathrm{m}$ can be viewed as the marginal probability for 
$\mathtt{X}$ and $\mathtt{Z}$ errors. 
%-----------------------------------------------------------

Let $M\times N$ be the size of the non-binary parity-check matrix $\HG$ over $\GF(2^p)$. 
The code length is $pN$ qubits. 
We deal with a $p$-bit sequence as a non-binary symbol which is simply referred to as symbol. 
Moreover, we deal with the symbol interchangeably as a symbol in $\GF(2^p)$. 

Note that the channel is the normal depolarizing channel. 
We assume the decoder knows the depolarizing probability $3f_\mathrm{m}/2$. 
For each row $m=1, \dotsc, M$  in $\HG$, 
let $N_m$ be the set of the non-zero entry indices in the $m$-th row. To be precise, $N_m:=\{n\mid \gamma_{m, n}\neq 0\}$. 
The decoder is given the syndrome symbols $\underline{s}_m\in \GF(2)^p$ for $m=1, \dotsc, M.$
To be precise, the decoder does not know the flipped qubits but their syndromes:
\begin{align}
 &\underline{s}_m=\sum_{n\in N_m}A(\gamma_{m, n})\underline{y}_n,\label{221113_6Jul10} 
\end{align} 
where $A$ is the isomorphism defined in Section \ref{223910_10Jun10} and $\underline{y}_n\in\GF(2)^p$ is a $p$-bit sequence corresponding to the $n$-th $p$-qubit sequence of flipped $pN$ qubits.

For simplicity, we concentrate on the decoding algorithm for $C$, since the decoding algorithm for $D$ is given 
by replacing $\Gamma$ with $\Delta$, and $A(\cdot)$ with $A^\mathsf{T}(\cdot)$ in the following algorithm. 
%-----------------------------------------------------------

\noindent{\bf The decoding algorithm of $C$}
\\\initialization:\\
For each column $n=1, \dotsc, N$  in $\HG$, 
let $M_n$ be the set of the non-zero entry indices in the $n$-th column. To be precise, $M_n:=\{m\mid \gamma_{m, n}\neq 0\}$. 
For each column $n$ in $\HG$ for $n=1, \dotsc, N$, calculate the initial probability $p_{n}^{(0)}(\underline{e})$ as follows. 
 \begin{align*}
  p_{n}^{(0)}(\underline{e})&= \Pr(\underline{e}_{n}=\underline{e}|\underline{Y}_{n}=\underline{0})=f_\mathrm{m}^{W_\mathrm{H}(\underline{e})}(1-f_\mathrm{m})^{p-W_\mathrm{H}(\underline{e})}\label{181252_11Jan10}
 \end{align*} 
 for  $\underline{e}\in \GF(2)^p$, where $f_\mathrm{m}$ is the flip probability of the channel and $W_\mathrm{H}(\underline{e})$ is the Hamming weight of $\underline{e}$. 
 For each column $n=1, \dotsc, N$ in $\HG$, copy the initial message $p_{nm}^{(0)}=p_{n}^{(0)}\in [0, 1]^{2^p}$ for $m\in M_n$. 
 Set the iteration round as $\ell:=0$. 
% -   -   -   -   -   -   -   -   -   -   -   -   
\\\\\ctov:\\ 
Each row $m$ has $L$ incoming messages $p_{nm}^{(\ell)}$ for $v\in N_m$.
The $m$-th row sends the following message ${q}^{(\ell+1)}_{mn}\in[0, 1]^{2^p}$ to each column $n\in N_m$. 
\begin{align}
&\tilde{p}^{(\ell)}_{nm}(\underline{e}) = {p}^{(\ell)}_{nm}(A(\gamma_{nm}^{-1})\underline{e}) \text{ for $\underline{e}\in\GF(2)^p$}, \\
&\tilde{q}^{(\ell+1)}_{mn}= \bm{1}_{\underline{s}_m}\bigotimes_{n'\in N_m\backslash\{n\}}\tilde{p}^{(\ell)}_{n'm}, \nonumber\\
&{q}^{(\ell+1)}_{mn}(\underline{e}) = \tilde{q}^{(\ell+1)}_{mn}(A(\gamma_{nm}) \underline{e})\text{ for $\underline{e}\in\GF(2)^p$}.
 \end{align} 
where 
$\bm{1}_{\underline{s}_m}$ is a probability on  $\GF(2)^{p}$ such that $\bm{1}_{\underline{s}_m}(\underline{e})=1$ for $\underline{e}=\underline{s}_m$ and 0 otherwise, and 
$q_1\otimes q_2\in[0, 1]^{2^p}$ is a convolution of $q_1\in[0, 1]^{2^p}$ and $q_2\in[0, 1]^{2^p}$. To be precise, 
\begin{equation*}
  (q_1\otimes q_2)(\underline{e}) = \sum_{\substack{\underline{f}, \underline{g}\in\GF(2)^p\\\underline{e}=\underline{f}+\underline{g}}}{q_1(\underline{f})q_2(\underline{g})} \text{ for $\underline{e}\in\GF(2)^p$}.
\end{equation*}
%The convolution seems the most complex part of the decoding. 
The convolutions are efficiently calculated via FFT and IFFT \cite{4155118}, \cite{RaU05/LTHC}.
Increment the iteration round as $\ell:=\ell+1$. 
\\
% -   -   -   -   -   -   -   -   -   -   -   -   
\\\vtoc: \\
Each column $n=1, \dotsc, N$ in $\HD$ has $J=2$ non-zero entries.
Let $M_n$ be the set of the column indices of the non-zero entry.
The message $p^{(\ell)}_{nm}\in [0, 1]^{2^p}$ sent from $n$ to $m\in M_n$ is given by
\begin{align*}
p^{(\ell)}_{nm}(\underline{e}) = \xi q_n^{(0)}(\underline{e})\prod_{m'\in M_n\backslash\{m\}}q^{(\ell)}_{m'n}(\underline{e}) \text{ for $x\in\GF(2)^p$}, 
 \end{align*} 
where $\xi$ is the normalization factor so that $\sum_{\underline{e}\in\GF(2)^p}p^{(\ell)}_{nm}(\underline{e})=1$. \\
% -   -   -   -   -   -   -   -   -   -   -   -   
\tentative: \\
For each $n=1, \dotsc, N$, the tentatively estimated $v$-th transmitted symbol is given as 
\begin{align*}
\hat{\underline{e}}_n^{(\ell)}=\mathop{\mathrm{argmax}}_{\underline{e}\in \GF(2)^p}p_n^{(0)}(\underline{e})\prod_{m\in M_n}\textstyle q^{(\ell)}_{mn}(\underline{e}). 
\end{align*} 
If $(\hat{\underline{e}}_0, \dotsc, \hat{\underline{e}}_N)$ has the same syndrome as $(\underline{s}_1, \dotsc, \underline{s}_M)$ which is defined in \eqref{221113_6Jul10}, in other words, 
for $m=1, \dotsc, LP$, 
\begin{align*}
 \sum_{n\in N_m}A(\gamma_{mn})\hat{\underline{e}}_n^{(\ell)}=\hat{\underline{s}}_m\in \GF(2)^p 
\end{align*}
for all $c=1, \dotsc, M$, the decoder outputs $(\hat{\underline{e}}_0, \dotsc, \hat{\underline{e}}_N)$ as the estimated error. 
Otherwise, repeat the latter 3 decoding steps. If the iteration round $\ell$ reaches a pre-determined 
number, the decoder outputs {\tt FAIL}.

Note that, in this algorithm,  the correlations between $\mathtt{X}$ errors and $\mathtt{Z}$ errors are neglected. 
In \cite[Section VI, C]{1337106} MacKay et al. used the knowledge about the channel properties for decoding, which improved the decoding performance. 
The most complex part of the decoding is the horizontal step, which requires $O(Nq\log(q))$ multiplications and additions when calculated via FFT, where $q=2^p$. 
\begin{figure*}[t]
\begin{center}
   \includegraphics[width=\textwidth]{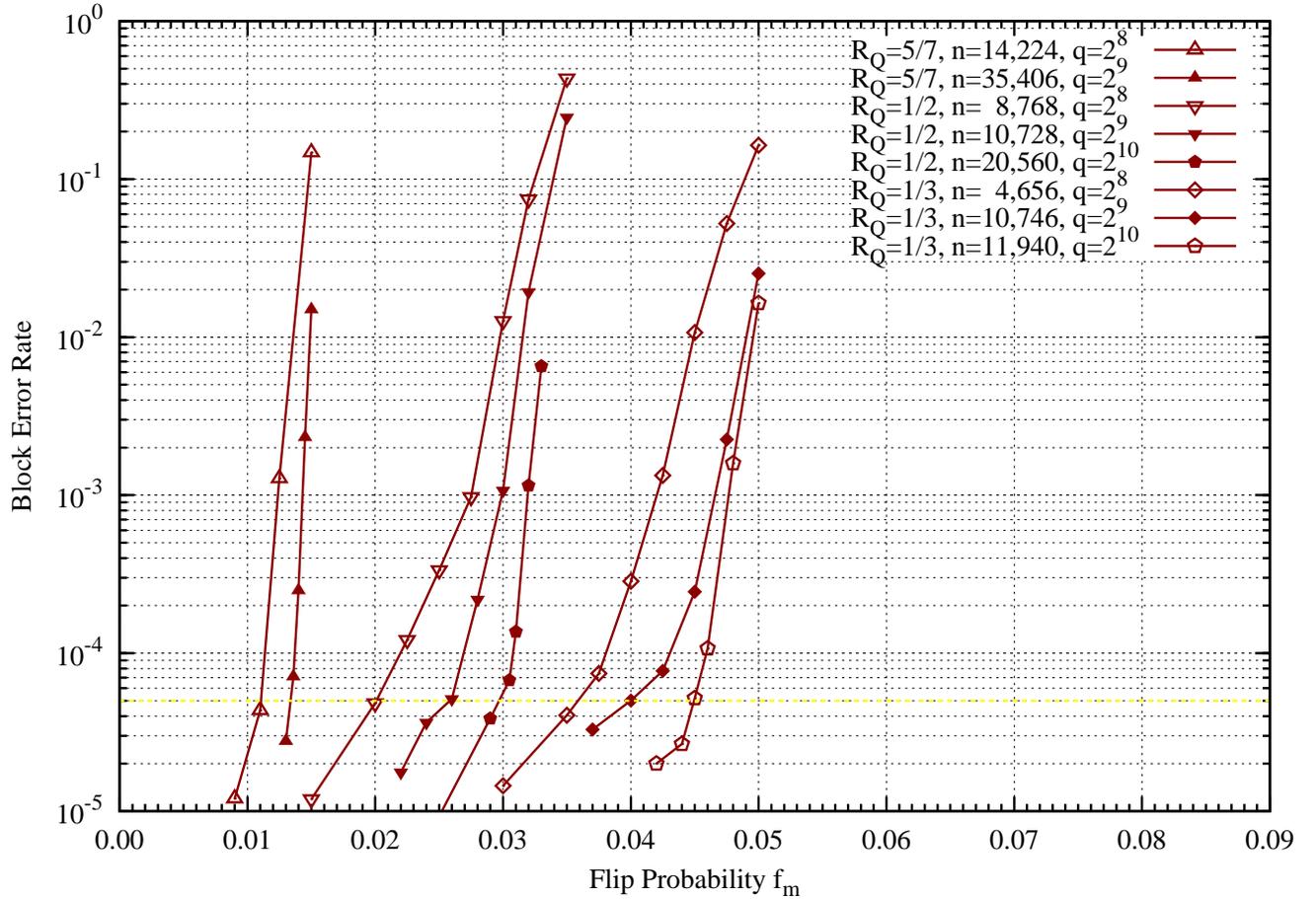}
 \caption{The block error probability of the constituent codes $C$ and $D$ of 
the proposed CSS code pair $(C, D)$ over the depolarizing channel with marginal flip probability $f_\mathrm{m}$ of $\mathtt{X}$ and $\mathtt{Z}$ errors. 
These codes are defined over $\GF(q)$ for $q=2^8, 2^9, 2^{10} $ and have quantum rate $R_\mathrm{Q}=1/3, 1/2, 5/7$. The code length is $n$ qubits.}
 \label{121852_13Jun10}
\end{center}
\end{figure*}
%-----------------------------------------------------
\begin{figure*}[t]
\begin{center}
   \includegraphics[width=\textwidth]{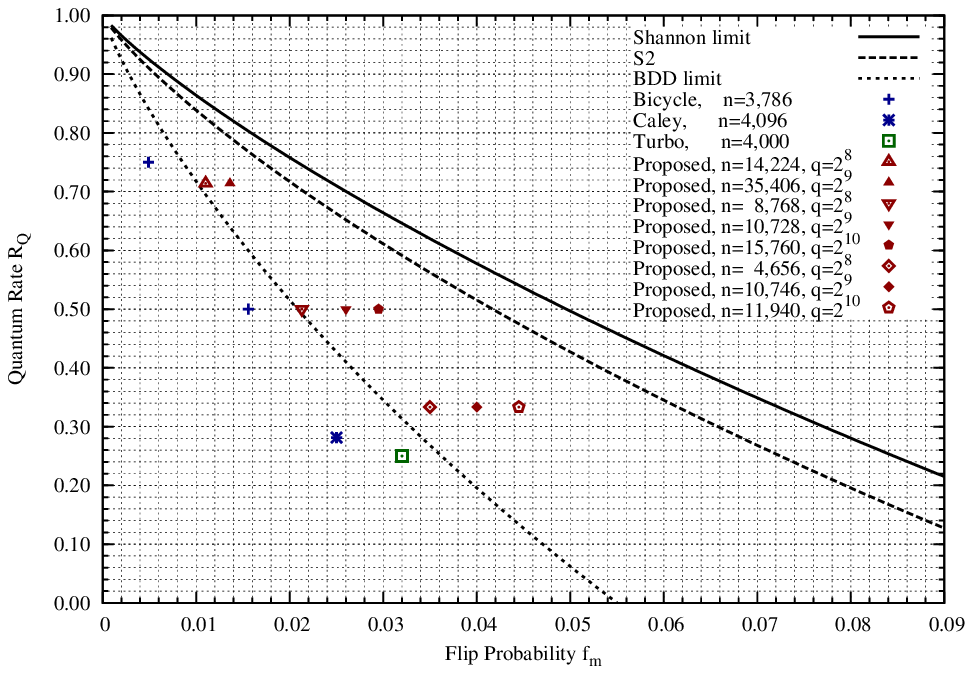}
 \caption{The performance of the proposed CSS code pair $(C, D)$ compared with the best CSS codes so far from \cite{1337106}, \cite{MoreSparseGraph} and  \cite{4957637} over the depolarizing channel with marginal flip probability $f_\mathrm{m}$ of $\mathtt{X}$ and $\mathtt{Z}$ errors. Each point is plotted at which the block probabilities of both two constituent codes are  $5\times 10^{-5}$. The block probability of the entire CSS code is $1-(1-5\times 10^{-5})^2\sim 10^{-4}$. 
The Shannon limit  of the depolarizing channel: $R_\mathrm{Q}=1-h(3f_\mathrm{m}/2)-3f_\mathrm{m}/2\log_2(3)$, where $h(\cdot)$  is the binary entropy function.
The curve labelled S2 is the achievable quantum rate if the correlations between $\mathtt{X}$ errors and $\mathtt{Z}$ errors are neglected: $R_\mathrm{Q}=1-2h(f_\mathrm{m})$. 
The curve labeled BDD is the performance limit when the bounded distance decoder is employed and the correlations between $\mathtt{X}$ errors and $\mathtt{Z}$ errors are neglected. $R_\mathrm{Q}=1-2h(2f_\mathrm{m})$.
The code length is $n$ qubits. The proposed codes are defined over $\GF(q)$.}
 \label{151306_29Jun10}
\end{center}
\end{figure*}
%%%%%%%%%%%%%%%%%%%%%%%%%%%%%%%%%
\section{Numerical Result}
\label{231523_10Jun10}
%%%%%%%%%%%%%%%%%%%%%%%%%%%%%%%%%
In this section, we demonstrate the proposed CSS code pair decoded by  the algorithm described in the previous section. 
The proposed CSS code pair $(C, D)$ is constructed as follows. 
First, by Theorem \ref{Hagiwara-Imai}, construct $JP\times LP$ binary matrices $\hHC$ and $\hHD$ with $J=2, L$, and $P$. 
Secondly, by the scheme described in Section \ref{230852_10Jun10}, construct $JP\times LP$w non-binary matrices $\HG$ and $\HD$ over $\GF(2^p)$. 
Finally, by the scheme described in Section \ref{223910_10Jun10}, we have $pJP\times pLP$ binary matrices $H_C$ and $H_D$. 
Thus, we obtain $C$ and $D$ are defined by the parity-check matrices $H_C$ and $H_D$, respectively. Note that $C$ and $D$ can not only be viewed as binary codes defined by $H_C$ and $H_D$ but also be viewed 
as non-binary codes defined by $\HG$ and $\HD$. 
The code length of the proposed CSS code is given as $n=pLP$ qubits or equivalently $LP$ symbols. 
The quantum rate $R_\mathrm{Q}$ of the proposed CSS code is given as 
\begin{align*}
 R_\mathrm{Q}=1-2J/L. 
\end{align*}

Fig.~\ref{121852_13Jun10} shows
the block error probability of the constituent codes $C$ and $D$ of the proposed CSS code pair $(C, D)$ over the depolarizing channel with marginal flip probability $f_\mathrm{m}$ of $\mathtt{X}$ and $\mathtt{Z}$ errors. 
Parameter are chosen $J=2$, $L=6, 8$ and  14 for $R_\mathrm{Q}=1/3, 1/2$ and  5/7,  respectively. 
The depolarizing probability is given by $3f_\mathrm{m}/2$.
The correlations between $\mathtt{X}$ errors and $\mathtt{Z}$ errors are neglected. 
Due to the symmetry of construction of $C$ and $D$, the block error probability of the constituent codes $C$ and $D$ are almost the same, hence we plot the block error probability of either $C$ or $D$. 
It is observed that for fixed $q=2^p$ and $R_\mathrm{Q}$, the codes with larger code length tend to have higher error floors. 
This is due to the fact that the proposed codes have poor minimum distance which is upper-bounded by $pL$. 
The error floors can be  improved by using larger $p$, i.e., larger field $\GF(2^p)$, which leads to the requirements of more complex decoding computations $O(Nq\log(q))$, where $q=2^p$. 

Fig.~\ref{151306_29Jun10} compares the proposed quantum codes with the best quantum codes so far. 
The horizontal axis is the flip probability at which the block error probability of one of the constituent classical code is $0.5\times 10^{-4}$. 
The vertical axis is the quantum rate $R_\mathrm{Q}$ of quantum codes. 
Since the proposed CSS codes have constituent classical codes $C$ and $D$ of the same classical rate $R_\mathrm{C}=1-J/L$, the quantum rate $R_\mathrm{Q}$ is given as $R_\mathrm{Q}=2R_\mathrm{C}-1=1-2J/L$.
It can be seen that the proposed codes outperform the state-of-the-art codes. 
In fact, the proposed codes surpass the BDD curve which is the limit of the bounded distance decoder, while the other codes fall inside the BDD curve. 

%----------------------------------------------------
%%%%%%%%%%%%%%%%%%%%%%%%%%%%%%%%%
\section{Conclusion}
%%%%%%%%%%%%%%%%%%%%%%%%%%%%%%%%%
We proposed a novel construction method of CSS codes. The resulting CSS codes can be viewed as non-binary LDPC codes over $\GF(2^p)$. 
Due to the sparse representation of the parity-check matrices, the proposed codes are efficiently decoded. 
The simulation results over the depolarizing channels show that the proposed codes outperform 
the other quantum error correcting codes which exhibited the best decoding performance so far. 
The error floors are lowered by increasing the size of the underlying Galois field, i.e., $2^p$. 
%%%%%%%%%%%%%%%%%%%%%%%%%%%%%%%%%
\section*{Acknowledgment}
%%%%%%%%%%%%%%%%%%%%%%%%%%%%%%%%%
The first author is grateful to R.~Matsumoto for useful discussions. 
The authors wish to thank the anonymous reviewer for his or her helpful comments which improved 
presentation of this paper.
%%%%%%%%%%%%%%%%%%%%%%%%%%%%%%%%%
\bibliographystyle{IEEEtran}
\bibliography{IEEEabrv,../../kenta_bib}
\end{document}